\documentclass[11pt]{article}

\usepackage[preprint,nonatbib]{neurips_2020}

\usepackage{lmodern}

\usepackage[utf8]{inputenc} 
\usepackage[T1]{fontenc}    
\usepackage{hyperref}       
\usepackage{url}            
\usepackage{booktabs}       
\usepackage{amsfonts}       
\usepackage{nicefrac}       
\usepackage{microtype}      
\usepackage{amsmath}
\usepackage{bm}
\usepackage{esvect}
\usepackage{amsthm}
\usepackage[noend,linesnumbered]{algorithm2e}
\usepackage{amssymb}
\usepackage{bbm}
\usepackage{enumitem}
\usepackage{graphicx}
\usepackage{caption}
\usepackage{subcaption}
\usepackage{multirow}
\usepackage{floatflt}
\usepackage{color}
\newcommand{\metas}[1]{\mathcal{G}_{#1}}
\newcommand{\QQ}{\texttt{{\sc div++}}}

\newcommand{\QQQQ}{\texttt{{\sc sub+div}}}
\makeatletter
\newtheorem*{rep@theorem}{\rep@title}
\newcommand{\newreptheorem}[2]{%
\newenvironment{rep#1}[1]{%
 \def\rep@title{#2 \ref{##1}}%
 \begin{rep@theorem}}%
 {\end{rep@theorem}}}
\makeatother

\makeatletter
\newtheorem*{rep@proposition}{\rep@title}
\newcommand{\newrepproposition}[2]{%
\newenvironment{rep#1}[1]{%
 \def\rep@title{#2 \ref{##1}}%
 \begin{rep@proposition}}%
 {\end{rep@proposition}}}
\makeatother

\makeatletter
\newtheorem*{rep@lemma}{\rep@title}
\newcommand{\newreplemma}[2]{%
\newenvironment{rep#1}[1]{%
 \def\rep@title{#2 \ref{##1}}%
 \begin{rep@lemma}}%
 {\end{rep@lemma}}}
\makeatother

\newtheorem{definition}{Definition}

\newtheorem{theorem}{Theorem}

\newtheorem{proposition}{Proposition}
\newtheorem{lemma}{Lemma}

\newtheorem{conjecture}{Conjecture}

\newcommand{\floor}[1]{\lfloor #1 \rfloor}

\newreptheorem{theorem}{Theorem}
\newrepproposition{proposition}{Proposition}
\newreplemma{lemma}{Lemma}

        {\hspace*{\fill}$\Box$\par}

\DeclareMathOperator*{\argmax}{arg\,max}

\newcommand{\R}{\mathbb{R}}
\newcommand{\E}{\mathbb{E}}
\newcommand{\M}{\mathcal{M}}

\newcommand{\smooth}{\sigma}

\let\emptyset\varnothing

\title{ 
A Parameterized Family of Meta-Submodular Functions
}

\author{
    Mehrdad Ghadiri \\
  Georgia Tech\\
  \texttt{ghadiri@gatech.edu} \\
   \And
    Richard Santiago \\ 
    ETH Zurich \\
    \texttt{rtorres@ethz.ch}
    \And
    Bruce Shepherd \\ 
    UBC \\
    \texttt{fbrucesh@cs.ubc.ca}
}
\date{}
\allowdisplaybreaks
\begin{document}
\maketitle
 
\begin{abstract}

Submodular function maximization has found a wealth of new applications in machine learning models during the past years. The related supermodular maximization models (submodular minimization) also offer an abundance of applications, but they appeared to be highly intractable even under simple cardinality constraints.
Hence, while there are well-developed tools for maximizing a submodular function subject to a matroid constraint, there is much less work on the corresponding supermodular maximization problems.

We give a broad parameterized family of monotone functions which includes submodular functions and a class of supermodular functions containing diversity functions. Functions in this parameterized family are called \emph{$\gamma$-meta-submodular}. We develop local search algorithms with approximation factors that depend only on the parameter $\gamma$. We show that the $\gamma$-meta-submodular families include well-known classes of functions such as meta-submodular functions ($\gamma=0$), metric diversity functions and proportionally submodular functions (both with $\gamma=1$), diversity functions based on negative-type distances or Jensen-Shannon divergence (both with $\gamma=2$), and $\sigma$-semi metric diversity functions ($\gamma = \sigma$).

\end{abstract}

\section{Introduction}
\label{sec:intro}

In the past decades, the catalogue of algorithms  available to combinatorial optimizers has been substantially extended 
to new settings which allow submodular objective functions.
These  developments in submodular maximization were occurring at the same time that researchers found a wealth of new applications in machine learning and data mining for these models \cite{kempe2003maximizing,krause2007near, boykov2001interactive,kohli2009p3,jegelka2011submodularity, lin2011class,streeter2009online,liu2013submodular,prasad2014submodular,dey2012contextual}. 

The related supermodular maximization models  (submodular minimization) also offer an abundance of applications, but they appeared to be highly intractable even under simple cardinality constraints \cite{svitkina2011submodular}. The applications include, but are not limited to, feature selection~\cite{ZadehGMZ17,ghadiri2019distributed}, neural architecture search~\cite{abs-1910-00370}, document aggregation~\cite{AbbassiMT13}, web search~\cite{AngelK11,YuLA09}, keyword search in databases~\cite{ZhaoZTC11}.

In some cases constrained supermodular maximization admits a constant factor approximation.
One such example arises in the realm of diversity maximization. Let $[n]=\{1,\ldots,n\}$ be our ground set and $A$ be a pairwise dissimilarity measure on the elements of $[n]$, where $A$ is a symmetric, zero-diagonal matrix with non-negative entries. Given an integer $r$, the goal of the diversity maximization problem is to find a set $S\subseteq [n]$ of size $r$ that maximizes $f(S)=\frac{1}{2}\sum_{i,j\in S}A(i,j)$. When $A$ is a metric distance (i.e., $A(i,j)\leq A(i,k)+A(k,j)$ for any $i,j,k\in [n]$), this problem admits a $2$-approximation~\cite{hassin1994notes,AbbassiMT13,BorodinLY12} which is tight~\cite{BhaskaraGMS16,BorodinLY12}. One might think that this is because of the nice pairwise structure of these functions. However,  when $A$ is not metric and its entries are from $\{0,1\}$, then this problem is equivalent to the densest $k$-subgraph problem
whose approximation is $O(n^{0.25+\epsilon})$~\cite{bhaskara2012polynomial}.  
In this case it cannot admit  a constant-factor approximation under the generally accepted complexity assumption of ETH \cite{manurangsi2017almost}.  In fact, the metric property is key to why diversity functions behave nicely in the former case.  This is generalized to the case where $A$ is a $\gamma$-semi-metric (i.e., $A(i,j)\leq \gamma(A(i,k)+A(k,j))$ for any $i,j,k\in [n]$). Namely, it is shown that maximizing a diversity function with a $\gamma$-semi-metric distance, subject to a cardinality constraint $|S|\leq r$, admits a $2\gamma$-approximation~\cite{ZadehG15}, which is tight \cite{ghadiri2019beyond}.

As  discussed, for a fixed   semi-metric parameter there is a constant-factor approximation for  functions with a pairwise structure.  Can this  parameter be generalized  to   general set functions?  We answer this question affirmatively in this paper. In order to define this generalization we introduce the following notation.

\begin{definition}
Let $f:2^{[n]}\rightarrow \mathbb{R}_{\geq  0}$ be a set function defined on the powerset of $[n]$. For a set $S\subseteq [n]$ and elements $i,j\in[n]$, we define the first-order difference (or marginal gain) of $i$ with respect to $S$ as
\[
B_i(S) := f(S+i) - f(S-i),
\]
where $S+i=S\cup\{i\}$ and $S-i=S\setminus\{i\}$. We also define the second-order difference of $i,j$ with respect to $S$ as
\begin{align*}
A_{ij}(S) & := B_j(S+i)-B_j(S-i) \\ & =f(S+i+j)-f(S+i-j)-f(S-i+j)+f(S-i-j).
\end{align*}
\end{definition}

Note that $A_{ij}(S)=A_{ji}(S)$. For the {\em diversity function} $f(S)=\frac{1}{2}\sum_{u,v\in S} A(u,v)$, we have $B_i(S)=\sum_{u\in S-i} A(i,u)$ and $A_{ij}(S)=A(i,j)$. The latter means that $A_{ij}(S)$ is constant for these diversity functions. Note that $B_i$ and $A_{ij}$ are defined for any set function and they do not need a pairwise structure. One can easily verify that $f$ is monotone if and only if $B_i(S)\geq 0$ for all $i$ and $S$. Moreover $f$ is submodular (supermodular) if and only if $A_{ij}(S)\leq 0$ ($A_{ij}(S)\geq 0$) for all $i,j$ and $S$. Now we can define our parameterized family of functions.

\begin{definition}
Let $\gamma\geq 0$. We say a set function $f$ is $\gamma$-meta-submodular ($\gamma$-MS) if, for any nonempty $S\subseteq [n]$ and $i,j\in [n]$, we have
\begin{align}
\label{eqn:meta}
A_{ij}(S)\leq \gamma \cdot \frac{B_i(S)+B_j(S)}{|S|}.
\end{align}
\end{definition}

For $\gamma=0$, our definition implies that $A_{ij}(S)\leq 0$ for any $i,j\in [n]$ and nonempty $S$. This is equivalent to the class of meta-submodular functions defined by Kleinberg et al~\cite{KleinbergPR98}. They defined this class of functions inspired by segmentation problems. Trivially, the class of $0$-MS functions contain all submodular functions.
For $\gamma$-semi-metric diversity functions, if $i,j\notin S$, the inequality in (\ref{eqn:meta}) is equivalent to
$
A(i,j)\leq \gamma (\sum_{k\in S} [A(i,k)+A(j,k)])/|S|.
$
This holds because for any $i,j,k \in [n]$, the $\gamma$-semi-metric property implies $A(i,j)\leq \gamma (A(i,k)+A(j,k))$. Therefore the above is just an average over such inequalities. Moreover the above inequality holds for $\gamma$-semi-metric diversity functions, regardless of whether $i,j$ are in $S$. Hence these functions are $\gamma$-MS --- see Proposition \ref{prop:semimetricmeta} in Appendix~\ref{app:metasubfamily}.

Negative-type distances and Jensen-Shannon divergence are among the most important distance functions. These distances are $2$-semi-metric (see \cite{ghadiri2019beyond}) and therefore, the diversity functions defined on them are $2$-MS. Another important class of functions are proportionally submodular functions which contains the functions that are the sum of a monotone submodular function and a metric diversity function~\cite{borodin2015proportionally}. Proportionally submodular functions are contained in the class of $1$-MS functions --- see Proposition~\ref{prop:borodinweakmeta} in Appendix~\ref{app:metasubfamily}.

As discussed, even for small $\gamma$, the class of $\gamma$-MS functions contain many important classes of functions used in machine learning and data mining applications. Moreover if $f,g$ are $\gamma$-MS and $\alpha>0$ is a real number, then $f+g$ and $\alpha f$ are also $\gamma$-MS. This allows combining $\gamma$-MS functions in different ways. We primarily focus on monotone functions and we denote by $\metas{\gamma}$ the family of non-negative, monotone set functions which are  $\gamma$-MS. Note that this implies that the $B_i$'s
are non-negative. Therefore one can see that $\metas{\gamma} \subseteq \metas{\gamma'}$ if $\gamma < \gamma'$.

In this work we consider the problem of maximizing a monotone $\gamma$-MS function subject to a matroid constraint. Before  discussing our results, we review some background material.


\subsection{Background, Notation, and Preliminary Results}

We need the following notation and definitions to explain our techniques and results. We use $[n]:=\{1,\ldots,n\}$ to refer to the ground set of a set function. For a set $R\subseteq [n]$, we denote by $\mathbbm{1}_R$ its characteristic vector. For $x=(x_1,\ldots,x_n)\in [0,1]^n$, $p_x(R)$ denotes the probability of picking set $R$ with respect to vector $x$. In other words,
$
p_x(R) = \prod_{v\in R}x_v \prod_{v\in [n]\setminus R} (1-x_v)
$.
The {\em  multilinear extension} of a set function $f:2^{[n]}\rightarrow \mathbb{R}$ is $F:[0,1]^n\rightarrow \mathbb{R}$, where
\[
F(x)=\sum\nolimits_{R\subseteq [n]} f(R)p_x(R) = \E_{R \sim x}[f(R)].
\]

\noindent
One can easily check that $f(R)=F(\mathbbm{1}_{R})$, $B_i(R)=\nabla_i F(\mathbbm{1}_{R})$, and $A_{ij}(R)=\nabla_{ij}^2 F(\mathbbm{1}_{R})$ --- see \cite{vondrak2008optimal}.
The following lemma describes the connection between the terms $A_{ij}$ and $B_i$ (see Appendix~\ref{app:preliminaries} for proof details).

\begin{lemma}[Discrete integral]
Let $f:2^{[n]} \to \R$, $i\in [n]$, and $R=\{v_1,\ldots, v_r \}\subseteq [n]$. Moreover, let $R_m=\{v_1,\ldots, v_m \}$ for $1\leq m\leq r$ and $R_0=\emptyset$. Then
$
B_i(R) = f(\{i\}) + \sum_{j=1}^{r} A_{i v_j} (R_{j-1}).
$
\label{lemma:discreteIntegral}
\end{lemma}
\noindent
We use $x^{\top}$ to denote the transpose of vector $x$. For vectors $x, y$, we denote the entrywise maximum of them by $x\vee y$, i.e., $z=x\vee y$ is a vector such that $z_i=\max\{x_i,y_i\}$.

A pair $\mathcal{M}=([n],\mathcal{I})$, where $\mathcal{I}$ is a family of subsets of $[n]$, is a matroid if: 1) for any $S\subseteq T\subseteq [n]$, if $T\in\mathcal{I}$ then $S\in\mathcal{I}$ (hereditary property); and 2) for any $S,T\in \mathcal{I}$, if $|S|<|T|$, then there exists $i\in T\setminus S$ such that $S+i\in \mathcal{I}$ (exchange property)~\cite{schrijver2003combinatorial}. 
We call $\mathcal{I}$ the set of independent sets of the matroid $\mathcal{M}$. Therefore given a $\gamma$-MS function $f$ and a matroid $\mathcal{M}=([n],\mathcal{I})$, our problem of interest is to find a set $S\in\mathcal{I}$ that maximizes $f(S)$.

A maximal independent set of a matroid is called a \emph{base}. All the bases of a matroid have the same size. The \emph{rank} of a matroid $\mathcal{M}$, denoted by $r$, is the size of a base of $\mathcal{M}$. Any subset of $[n]$ not in $\mathcal{I}$ is called a dependent set of $\mathcal{M}$.
A minimal dependent set of a matroid is called a \emph{circuit}. Note that the size of circuits are not necessarily equal. We usually denote the size of the smallest circuit of $\mathcal{M}$ by $c$ ($=c(\mathcal{M})$).

Two important families of matroids are uniform matroids and graphic matroids. Given an integer $r$, the set of independent sets of a uniform matroid is $\mathcal{I}=\{S\subseteq [n]: |S|\leq r\}$. Therefore cardinality constraints are a special class of matroid constraints. For a uniform matroid, it is not hard to see that the rank is $r$ and the size of the smallest circuit is $c=r+1$. Given a graph $G=(V,E)$, the graphic matroid on $G$ is $\mathcal{M}=(E,\mathcal{I})$ where $\mathcal{I}$ is the set of all forests of $G$. If $G$ is connected then $r=|V|-1$ and $c$ is the size of the smallest cycle of $G$. Matroids contain many more interesting family of constraints --- see \cite{schrijver2003combinatorial}. We frequently use the following result in our proofs.

\begin{lemma}[\cite{schrijver2003combinatorial}]
\label{lem:exchangematroid}
    Let $\mathcal{M}=([n],\mathcal{I})$ be a matroid and $S,T$ be two bases of $\mathcal{M}$. Then there exists a bijective mapping $g:S\setminus T\rightarrow T\setminus S$ such that $S-i+g(i)\in\mathcal{I}$ for any $i\in S\setminus T$. 
\end{lemma}


\subsection{Our Results}

Recall that $\mathcal{G}_\gamma$ denotes the family of non-negative, monotone set functions which are $\gamma$-meta submodular.
Our most general result states that for these functions, there is an approximation factor which depends only on $\gamma$. 
We remark that for
constant values of $\gamma$ we obtain a new tractable (parameterized) class of functions.

\begin{theorem}
\label{thm:LS-general}
Let $f \in \metas{\gamma}$.
Then a local search algorithm 
gives an
$O(\gamma^2 2^{4\gamma})$-approximation for maximizing $f$ subject to a matroid constraint.
\end{theorem}

One can improve the above approximation by requiring
additional assumptions on the function $f$. The following result shows that
if the corresponding $B_i$'s are submodular, then the exponential factor from Theorem~\ref{thm:LS-general} improves to a quadratic factor in terms of $\gamma$. We remark that submodularity of the $B_i$’s is just the notion of second-order submodularity introduced in~\cite{korula2018online}, and is also equivalent to the non-positivity of the third-order partial derivatives of the multilinear extension. Note that it is also equivalent to having $A_{ij}(S+k)-A_{ij}(S-k)\leq 0$ for all $i,j,k,S$.

\begin{theorem}
\label{thm:LS-2nd-order}
Let $f \in \metas{\gamma}$ such that $f$ is also second-order submodular
(that is, $B_i$'s are submodular). 
Let $\M$ be a matroid of rank $r$ that has the smallest circuit size of $c$. Then the modified local search algorithm (Algorithm~\ref{alg:localsearch}) gives an $O(\gamma+\frac{\gamma^2}{r})$-approximation for maximizing $f$ subject to $\M$. If in addition $f$ is supermodular, then this can be further improved to an $O(\min\{\gamma + \frac{\gamma^2}{r}, \frac{\gamma r}{c-1} \}) \leq O(\gamma^{3/2})$-approximation.
\end{theorem}

As we discussed $\gamma$-semi-metric diversity functions are $\gamma$-MS. One can easily check that such diversity function are also supermodular and second-order submodular. The reason is that for any $i,j,k,S$, we have $A_{ij}(S+k)-A_{ij}(S-k)=A(i,j)-A(i,j)=0$ and $A_{ij}(S)=A(i,j)\geq 0$. Therefore,  Theorem~\ref{thm:LS-2nd-order} guarantees an $O(\gamma^{3/2})$-approximation for maximizing a $\gamma$-semi-metric diversity function subject to a matroid constraint. This matches the current best known approximation for this problem given in \cite{ghadiri2019beyond}. The latter uses a continuous relaxation approach, which involves solving a continuous optimization problem and rounding the fractional solution to an integral one.
We remark that while the $O(\gamma^{3/2})$-approximation given in \cite{ghadiri2019beyond} only applies to $\gamma$-semi-metrics, our result holds for a larger class of functions. 
That is, for the class of supermodular, second-order submodular, $\gamma$-MS functions, which does not necessarily have the nice pairwise structure of $\gamma$-semi-metrics.
Morevover our algorithm is a simple combinatorial algorithm.

We note that for some matroid classes, the approximation factors in Theorem~\ref{thm:LS-2nd-order} are better than $O(\gamma^{3/2})$. For instance, uniform matroids (and more generally paving matroids) satisfy $c \geq r$. Hence the term $\frac{\gamma r}{c-1}$ gives a linear approximation of $O(\gamma)$.

\subsection{Techniques}

The class of $\gamma$-meta-submodular functions are closely related to the newly introduced concept of one-sided smoothness \cite{ghadiri2019beyond}. 
A continuously twice differentiable function $F:[0,1]^n\rightarrow\mathbb{R}$ is called \emph{one-sided $\smooth$-smooth} at $x \neq \vec{0}$ if for any $u\in[0,1]^n$,
\begin{equation}
\label{eqn:OSS}
\frac{1}{2} u^{\top}\nabla^2 F(x) u \leq \smooth\cdot (\frac{||u||_1}{||x||_1}) u^{\top} \nabla F(x).
\end{equation}

A function $F$ is {\em one-sided $\smooth$-smooth} if it is $\smooth$-smooth at any non-zero point of its domain.
It is shown in \cite{ghadiri2019beyond} that
the smoothness parameter governs the approximability of the associated continuous maximization problem $\max_{x\in P} F(x)$ where $P$ is a downwards closed polytope and $F$ is a monotone one-sided smooth function. 
Our first observation is that the one-sided smoothness of the multilinear extension of a set function $f$ implies the meta-submodularity of $f$ --- see Appendix \ref{app:smoothness} for proof details.

\begin{proposition}
\label{prop:smoothisinmeta}
Let $f$ be a set function and $F$ be its multilinear extension. If $F$ is one-sided $(\gamma/2)$-smooth, then $f$ is $\gamma$-MS.
\end{proposition}

In fact the $\gamma$-MS definition can be derived from one-sided $(\gamma/2)$-smoothness if we only consider (\ref{eqn:OSS}) for some specific $x$ and $u$. Suppose (\ref{eqn:OSS}) holds for $x=\mathbbm{1}_R$ and $u=\mathbbm{1}_{\{i,j\}}$. Then

 \[
A_{ij}(R) = \frac{1}{2} (2 u_i u_j \nabla^2 F_{ij}(x)) \leq \frac{\gamma }{2}\cdot \frac{u_i+u_j}{||x||_1} (u_i \nabla_i F(x) + u_j \nabla_j F(x)) = \gamma\cdot \frac{B_i(R)+B_j(R)}{|R|}.
 \]

\noindent
Conversely, if $f$ satisfies a  probabilistic version of (\ref{eqn:meta}), then $F$ is one-sided smooth (see Appendix~\ref{app:smoothness} for proof details). 

\begin{lemma}
\label{lemma:probabilisticVersion}
Let $f$ be a non-negative, monotone set function and $F$ be its multilinear extension. Let $x\in [0,1]^n$ and $\gamma\geq 0$. If for any $i,j\in [n]$ we have the following: 
\begin{equation}
\label{eq:expect-ineq}
\E_{R \sim x}[|R|]\cdot\E_{R \sim x}[A_{ij}(R)]\leq \gamma \cdot(\E_{R \sim x}[B_i(R)]+\E_{R \sim x}[B_j(R)]),
\end{equation}
where  $R\sim x$  denotes a random set that contains element $i$ independently with probability $x_i$,
then $F$ is one-sided $\gamma$-smooth at $x$.
\end{lemma}

 We call this probabilistic version the {\em expectation inequality}  (\ref{eq:expect-ineq}).
 We have proved this inequality holds (modulo a constant factor) in the supermodular case (see Lemma~\ref{lemma:expIneq} in Appendix~\ref{app:smoothness}).  This yields the following.
 \begin{theorem}
Let $f$ be a supermodular function such that $f \in \mathcal{G}_\gamma$. Then its multilinear extension $F$ is one-sided $(\max \{ 3\gamma,2\gamma+1 \})$-smooth.
\label{thm:supmetasmooth}
\end{theorem}

 We conjecture that for $\gamma>0$, the multilinear extension of any $\gamma$-meta-submodular function is one-sided $O(\gamma)$-smooth. We use one-sided smoothness to prove Theorem~\ref{thm:LS-general}.
While it is  most convenient to have the smoothness property for the multilinear extension $F$ at every point of its domain, in order to prove Theorem~\ref{thm:LS-general} we only need it on a subdomain of $F$. We prove the following ``subdomain smoothness'' property in Section~\ref{sec:local-search}.

\begin{theorem}
	\label{thm:metasubsmooth}
	Let $f \in \metas{\gamma}$  and $F$ be its multilinear extension. Let $\alpha\geq 1$ and $S\subseteq [n]$ be non-empty. Then $F$ is one-sided $\alpha\gamma$-smooth on $\{x|x\geq \mathbbm{1}_{S} \, , \, ||x||_1\leq \alpha|S|\}$.
\end{theorem}

\subsection{Additional  Related Work}

For metric diversity functions, there exists a $2$-approximation  subject to a cardinality constraint
\cite{RaviRT94,hassin1994notes}.
Moreover, this has been extended to the case of
matroid constraints 
\cite{AbbassiMT13,BorodinLY12}. A PTAS is recently given for maximizing diversity functions on negative-type distances subject to a matroid constraint \cite{CevallosEZ16,CevallosEZ17}. There exists
a $10.22$-approximation for maximizing proportionally submodular functions subject to a matroid constraint~\cite{BorodinLY14weak,borodin2015proportionally}. 

Other extensions of submodular functions with respect to some sliding parameter (measuring how close a set function is to being submodular) have been considered in the literature. These include the class of weakly submodular functions, introduced in \cite{das2011submodularicml} and further studied in \cite{elenberg2017streaming,khanna2017scalable,hu2014efficient,chen2017weakly,bian2017guarantees,santiago2020weakly}. The class of set functions with supermodular degree $d$ (an integer between $0$ and $n-1$ such that $d=0$ if and only if $f$ is submodular), introduced in \cite{feige2013welfare} and further considered in \cite{feldman2014constrained,feldman2017building}. This has been extended to the Supermodular Width hierarchy \cite{chen2018capturing}. The class of $\epsilon$-approximate submodular functions studied in \cite{horel2016maximization}. The hierarchy over monotone set functions introduced in \cite{feigeFIILS15}, where levels of the hierarchy correspond to the degree of complementarity in a given function. They refer to this class as MPH (Maximum
over Positive Hypergraphs), and MPH-k denotes the $k$-th level in the hierarchy where $1\leq k\leq n$. The highest level MPH-n of the hierarchy captures all monotone functions, while the lowest level MPH-1 captures the class of \emph{XOS} functions (which include submodular).

We remark that our class of $\gamma$-meta-submodular functions differs from all the above extensions, since, for instance, none of them captures the class of metric diversity functions (in the sense of having a parameter that gives a good, say $O(1)$, approximation) while ours does.

\section{A Modified Local Search Algorithm}

In this section we introduce the modified local search algorithm, i.e., Algorithm \ref{alg:localsearch}. The first part of the algorithm (steps 1-6) consists of the standard local search procedure, where an approximate local optimum set $S$ is found. A set $S$ is an $\epsilon$-approximate local optimum if for any $i\in S$ and $j\in [n]\setminus S$ that $S-i+j\in\mathcal{I}$, we have $f(S-i+j)\leq (1+\frac{\epsilon}{n^2})f(S)$. It is a standard practice to find an approximate local optimum instead of an actual local optimum as the latter might take exponential time. The new component of the algorithm consists of step 7,  which requires finding a maximum weighted bipartite matching with $\floor{\frac{c-1}{2}}$ edges in an auxiliary graph, in order to produce a second candidate solution $S'$ --- which is the node set of the matching. Note that $S'$ is an independent set of the matroid because its size is less than $c$, the minimum size of any circuit in the matroid. The algorithm then returns the better of the two solutions $S$ and $S'$. The new step (i.e., step 7) plays a key role in improving the approximation factor when the function is supermodular --- see Theorem \ref{thm:LS-2nd-order}.

The auxiliary graph is a complete weighted bipartite graph $G$ with node sets $S$ and $[n]\setminus S$. The edge weights are $w(i,j):= A_{ij}(S)$ for $i\in S$ and $j\in [n]\setminus S$. We want to find a maximum weighted matching with $\floor{\frac{c-1}{2}}$ edges in $G$. This matching can be found by a simple reduction to the maximum weighted bipartite matching problem as follows: add $|S|-\floor{\frac{c-1}{2}}$ dummy nodes to $[n]\setminus S$ and connect them to all the nodes in $S$ with a weight equal to the maximum of $w(i,j)$'s. Finding a maximum weighted bipartite matching in this graph is equivalent to finding a maximum weighted bipartite matching with $\floor{\frac{c-1}{2}}$ edges in the original graph. This matching can be found in time $O(n^2 (r+\log n))$ using the Hungarian algorithm with the Dijkstra algorithm and Fibonacci heap \cite{FredmanT87}.

We note that the standard local search algorithm (i.e., the one consisting of steps 1-6 of Algorithm \ref{alg:localsearch}) has been
previously used for maximizing a submodular \cite{nemhauser1978analysis,lee2010submodular} and diversity \cite{AbbassiMT13,ZadehGMZ17} objective functions subject to a matroid constraint. 

\RestyleAlgo{algoruled}
\begin{algorithm}[th]
	\footnotesize
	\textbf{Input:} A set function $f$, a matroid $\M=([n],\mathcal{I})$ with circuits of minimum size $c$, and $\epsilon>0$.\\[0.6ex]
	$S_0 \leftarrow \argmax_{\{v,v'\} \in \mathcal{I}} f(\{v,v'\})$\\
	$S \leftarrow$ any base of $\M$ that contains $S_0$\\
	\While{$S$ is not an approximate local optimum} {
		Find $i\in S$ and $j\in [n]\setminus S$ such that $S-i+j\in\mathcal{I}$ and $f(S-i+j)\geq (1+\frac{\epsilon}{n^2})f(S)$\\
		$S\leftarrow S-i+j$\\
	}
	Create a complete weighted bipartite graph $G$ with node sets $S$ and $[n]\setminus S$, and edge weights $w(i,j) := A_{ij}(S)$ for each $i \in S$ and $j \notin S$. Find a maximum weighted matching $M$ in $G$ of (edge) cardinality $\floor{\frac{c-1}{2}}$, and let $S'$ denote the node set of $M$.\\
	\Return $\argmax\{f(S),f(S')\}$\\
	
	\caption{Local search under matroid constraint}
	\label{alg:localsearch}
\end{algorithm}

\section{General $\gamma$-Meta-Submodular Functions}
\label{sec:local-search}

In this section we present the main algorithmic result for general monotone $\gamma$-meta-submodular functions. Our goal is to show that an approximate local optimum solution $S$ is a good approximation for a global optimum solution $T$. To prove this, we need to bound $f(T)$ by a factor of $f(S)$. Since $f$ is monotone, we know $f(T)\leq f(S\cup T)$. Therefore, instead of bounding $f(T)$ directly, we find a bound for $f(S\cup T)$. To do so, we can use the multilinear extension of $f$ and Taylor's expansion of the extension. Let $F$ be the multilinear extension of $f$. Then by Taylor's theorem, for some $\epsilon'\in [0,1]$, we have
\begin{align*}
f(S\cup T) & = F(\mathbbm{1}_S\vee \mathbbm{1}_T) = F(\mathbbm{1}_S+\mathbbm{1}_{T\setminus S}) = F(\mathbbm{1}_S)+\mathbbm{1}_{T\setminus S}^{\top}\nabla F(\mathbbm{1}_S+\epsilon' \mathbbm{1}_{T\setminus S}) \\ & = f(S)+\mathbbm{1}_{T\setminus S}^{\top}\nabla F(\mathbbm{1}_S+\epsilon' \mathbbm{1}_{T\setminus S}).
\end{align*}
So we only need to bound $\mathbbm{1}_{T\setminus S}^{\top}\nabla F(\mathbbm{1}_S+\epsilon' \mathbbm{1}_{T\setminus S})$ in terms of $f(S)$. To do so, we use a subdomain smoothness of meta-submodular functions and then we use this property to bound the mentioned term. Hence in this section, we first prove the subdomain smoothness of meta-submodular functions (Lemma~\ref{lemma:probabilisticVersion} and Theorem~\ref{thm:metasubsmooth}), and then we show some bounds on the directional derivative of the multilinear extension of meta-submodular function using the subdomain smoothness property (Lemma~\ref{lemma:sumGradEws} and Lemma~\ref{lemma:grad-ratio-ews}). We then use these bounds to prove that an approximate local optimum is a good approximation for a global optimum (Theorem~\ref{thm:LS-general}).

\begin{replemma}{lemma:probabilisticVersion}
Let $f$ be a non-negative, monotone set function and $F$ be its multilinear function. Let $x\in [0,1]^n$ and $\gamma\geq 0$. If for any $i,j\in [n]$ we have 
\[
\E_{R \sim x}[|R|]\cdot\E_{R \sim x}[A_{ij}(R)]\leq \gamma\cdot(\E_{R \sim x}[B_i(R)]+\E_{R \sim x}[B_j(R)]),
\]
or equivalently (see \cite{vondrak2008optimal} or Lemma~\ref{lemma:gradientbi} in Appendix~\ref{app:smoothness}),
\[
||x||_1\nabla_{ij}^2 F(x) \leq \gamma (\nabla_i F(x)+ \nabla_j F(x)),
\]
then $F$ is one-sided $\gamma$-smooth at $x$.
\end{replemma}
\begin{proof}
We have
\begin{align*}
u^{\top} \nabla^2 F(x) u & = \sum_{i=1}^n \sum_{j=1}^n u_i u_j \nabla_{ij}^2 F(x) \leq \frac{\gamma}{||x||_1} \sum_{i=1}^n \sum_{j=1}^n u_i u_j (\nabla_i F(x)+\nabla_j F(x))
\\ & = \frac{\gamma}{||x||_1}(\sum_{i=1}^n \sum_{j=1}^n u_i u_j \nabla_i F(x)+\sum_{i=1}^n \sum_{j=1}^n u_i u_j \nabla_j F(x))
\\ & = \frac{\gamma}{||x||_1}(\sum_{i=1}^n u_i \nabla_i F(x) (\sum_{j=1}^n u_j) +\sum_{i=1}^n u_i (\sum_{j=1}^n u_j \nabla_j F(x)))
\\ & = \frac{\gamma}{||x||_1}(||u||_1 \sum_{i=1}^n u_i \nabla_i F(x) + ||u||_1 \sum_{j=1}^n u_j \nabla_j F(x))
\\ & = 2\gamma \left(\frac{||u||_1}{||x||_1} \right) ( u^{\top} \nabla F(x)).
\end{align*}
\end{proof}

Now we can show the following subdomain smoothness property which will be used to bound the Taylor's polynomial of the multilinear extension of $\gamma$-MS functions.

\begin{reptheorem}{thm:metasubsmooth}
	Let $f \in \metas{\gamma}$  and $F$ be its multilinear extension. Let $\alpha\geq 1$ and $S\subseteq [n]$ be non-empty. Then $F$ is one-sided $\alpha\gamma$-smooth on $\{x|x\geq \mathbbm{1}_{S} \, , \, ||x||_1\leq \alpha|S|\}$.
\end{reptheorem}
\begin{proof}
Let $y\in\{x|x\geq \mathbbm{1}_{S} \, , \, ||x||_1\leq \alpha|S|\}$. First, we show that 
\[
||y||_1\nabla^2_{ij} F(y) \leq \gamma \alpha (\nabla_i F(y) + \nabla_j F(y)).
\]
We know $\nabla^2_{ij} F(y) = \sum_{R\subseteq [n]} A_{ij}(R) p_y(R)$. Since $y\geq \mathbbm{1}_S$, $p_y(R)=0$ for any $R$ that is not a superset of $S$. Therefore, $\nabla^2_{ij} F(y) = \sum_{R\subseteq [n]\setminus S} A_{ij}(S\cup R) p_y(S\cup R)$. We have 
\begin{align*}
||y||_1 \nabla^2_{ij} F(y) & = ||y||_1\sum_{R\subseteq [n]\setminus S} A_{ij}(S\cup R) p_y(S\cup R) \leq \alpha|S|\sum_{R\subseteq [n]\setminus S} A_{ij}(S\cup R) p_y(S\cup R) \\ & \leq \sum_{R\subseteq [n]\setminus S} \frac{\gamma \alpha|S|}{|S\cup R|} (B_i(S\cup R) + B_j(S\cup R)) p_y(S\cup R) \\ & \leq \sum_{R\subseteq [n]\setminus S} \gamma \alpha (B_i(S\cup R) + B_j(S\cup R)) p_y(S\cup R) \\ & \leq \gamma \alpha (\nabla_i F(y) + \nabla_j F(y)).
\end{align*}
\noindent
Now, by Lemma~\ref{lemma:probabilisticVersion}, we conclude that $F$ is one-sided ($\alpha\gamma$)-smooth at $y$.
\end{proof}

To analyse the local search algorithm, we use the following technical lemmas which use subdomain one-sided smoothness (Theorem~\ref{thm:metasubsmooth}) to bound the Taylor series expansion of the multilinear extension of $\gamma$-MS functions. 
\begin{lemma}
	\label{lemma:sumGradEws}
	Let $f\in \metas{\gamma}$ and $F$ be its multilinear extension. Let $R\subseteq [n]$ such that $|R|\geq 2$. Then
	\[
\mathbbm{1}_R^{\top} \nabla F(\mathbbm{1}_R)=\sum_{i \in R} B_i (R-i) \leq ((\frac{\lfloor \frac{|R|}{2} \rfloor^2 + \lceil \frac{|R|}{2} \rceil^2}{\lfloor \frac{|R|}{2} \rfloor \lceil \frac{|R|}{2} \rceil }+2)\gamma+2) f(R) \leq (5\gamma+2)f(R)
\]
\end{lemma}

\begin{proof}
Partition $R$ into two sets of size $\lfloor \frac{|R|}{2} \rfloor$ and of size $\lceil \frac{|R|}{2} \rceil$ like $S$ and $T$. Using Theorem~\ref{thm:metasubsmooth}, we know that $F$ is one-sided $((\lfloor \frac{|R|}{2} \rfloor / \lceil \frac{|R|}{2} \rceil +1)\gamma)$-smooth on $\{y|\mathbbm{1}_T\leq y\leq \mathbbm{1}_R\}$ and it is one-sided $((\lceil \frac{|R|}{2} \rceil / \lfloor \frac{|R|}{2} \rfloor +1)\gamma)$-smooth on $\{y|\mathbbm{1}_S\leq y\leq \mathbbm{1}_R\}$. Let $\alpha=(\lceil \frac{|R|}{2} \rceil / \lfloor \frac{|R|}{2} \rfloor +1)$. We show that 
\[
\sum_{i\in T} B_i(R-i) \leq \alpha\gamma f(R).
\]
Let $h(t)=F(\mathbbm{1}_S+t\mathbbm{1}_T)$ and $g(t)=\mathbbm{1}_T^{\top} \nabla F(\mathbbm{1}_S+t\mathbbm{1}_T)$ where $0\leq t\leq 1$. Note that $g(t)=h'(t)$ and $\mathbbm{1}_T^{\top} \nabla^2 F(\mathbbm{1}_S+t\mathbbm{1}_T) \mathbbm{1}_T = g'(t)$. Since $F$ is one-sided $\alpha\gamma$-smooth at any given point $\mathbbm{1}_S\leq y \leq \mathbbm{1}_R$, we have
\[
g'(t)=\mathbbm{1}_T^{\top} \nabla^2 F(\mathbbm{1}_S+t\mathbbm{1}_T) \mathbbm{1}_T \leq \alpha\gamma (\frac{||\mathbbm{1}_T||_1}{||\mathbbm{1}_S+t\mathbbm{1}_T||_1})(\mathbbm{1}_T^{\top} \nabla F(\mathbbm{1}_S+t\mathbbm{1}_T)) \leq \alpha\gamma\frac{1}{t} g(t).
\]
Therefore, $tg'(t)\leq \alpha\gamma g(t)$. Integrating both sides, we get
\[
\int_{0}^1 t g'(t) dt \leq \int_{0}^1 \alpha\gamma g(t) dt.
\]
Applying the integration by parts formula to the left hand side, we get
\[
t g(t)\biggl|_{0}^{1} - \int_{0}^1 g(t) dt \leq \alpha\gamma \int_{0}^1 g(t) dt.
\]
It follows that
\[
1 \cdot g(1) - 0 \cdot g(0) = \mathbbm{1}_T^{\top} \nabla F(\mathbbm{1}_S+\mathbbm{1}_T) = \mathbbm{1}_T^{\top} \nabla F(\mathbbm{1}_R) = \sum_{i\in T} B_i(R-i) \leq (\alpha\gamma +1) \int_{0}^1 g(t) dt.
\]
By using $g(t)=h'(t)$ we have
\begin{align*}
\sum_{i\in T} B_i(R-i) & \leq (\alpha\gamma+1) \int_{0}^1 h'(t) dt = (\alpha\gamma +1)  (h(1) -  h(0)) \\ & = (\alpha\gamma+1)(F(\mathbbm{1}_S+\mathbbm{1}_T) - F(\mathbbm{1}_S)) \\ & \leq (\alpha\gamma+1)F(\mathbbm{1}_R)=(\alpha\gamma +1)f(R).
\end{align*}
This means that 
\[
\sum_{i\in T} B_i(R-i)\leq ((\lceil \frac{|R|}{2} \rceil / \lfloor \frac{|R|}{2} \rfloor +1)\gamma+1) f(R).
\]
With the same argument we can conclude that
\[
\sum_{i\in S} B_i(R-i)\leq ((\lfloor \frac{|R|}{2} \rfloor / \lceil \frac{|R|}{2} \rceil +1)\gamma+1) f(R),
\]
and combining these inequalities yields the lemma.
\end{proof}

For our next result, we use the following lemma from \cite{ghadiri2019beyond} which bounds the directional derivative at points close to $x$ by a factor of the directional derivative at $x$.

\begin{lemma} [\cite{ghadiri2019beyond}]
Let $x \in [0,1]^n\setminus\{\vec{0}\}$, $u \in [0,1]^n$ and $\epsilon >0$ such that $x+\epsilon u \in [0,1]^n$. Let $F: [0,1]^n \rightarrow \mathbb{R}$ be a non-negative, monotone function which is one-sided $\sigma$-smooth on $\{y|x+\epsilon u \geq y\geq x\}$.
Then 
\[
u^{\top} \nabla F(x+\epsilon u) \leq \left(\frac{||x+\epsilon u||_1}{||x||_1}\right)^{2\sigma} (u^{\top} \nabla F(x)).
\]
\label{lemma:epsilonchangegradient}
\end{lemma}

The following is an immediate result of Theorem~\ref{thm:metasubsmooth} and Lemma~\ref{lemma:epsilonchangegradient}.

\begin{lemma}
	\label{lemma:grad-ratio-ews}
	Let $f \in \metas{\gamma}$ and $F$ be its multilinear function. Let $R\subset [n]$, and $x\in [0,1]^n$ such that $||x||_1\leq |R|$. Let $u=\mathbbm{1}_R\vee x-\mathbbm{1}_R$. Then for $0\leq \epsilon \leq 1$, we have
	$
	u^{\top} \nabla F(\mathbbm{1}_R+\epsilon u) \leq 2^{4\gamma} u^{\top} \nabla F(\mathbbm{1}_R)
	$
\end{lemma}
\begin{proof}
By Theorem~\ref{thm:metasubsmooth}, we know that $F$ is one-sided $2\gamma$-smooth on $A=\{y|y\geq  \mathbbm{1}_R,||y||_1\leq 2|R|\}$. Therefore $F$ is one-sided $2\gamma$-smooth on $B=\{y|\mathbbm{1}_R+\epsilon u\geq y\geq  \mathbbm{1}_R\}$ because $B\subseteq A$. Therefore, the desired result yields by Lemma~\ref{lemma:epsilonchangegradient}.
\end{proof}


We now prove Theorem \ref{thm:LS-general}.
We note that this result does not use the last step of Algorithm~\ref{alg:localsearch} where we find a maximum matching.
We discuss the runtime of Algorithm~\ref{alg:localsearch} for meta-submodular functions in Appendix~\ref{app:subseclocaltime}.

\begin{reptheorem}{thm:LS-general}
Let $f\in \metas{\gamma}$ and $\M=([n],\mathcal{I})$ be a matroid of rank $r$. Let $T\in \mathcal{I}$ be an optimum set, i.e.,
$
T \in \argmax_{R\in \mathcal{I}} f(R),
$
and $S\in \mathcal{I}$ be an $(1+\frac{\epsilon}{n^2})$-approximate local optimum, i.e., for any $i$ and $j$ such that $S-i+j\in \mathcal{I}$,
$
(1+\frac{\epsilon}{n^2})f(S)\geq f(S-i+j),
$
where $\epsilon>0$ is a constant. Then if $\gamma=O(r)$, $f(T)\leq O(\gamma 2^{4\gamma}) f(S)$ and if $\gamma=\omega(r)$, $f(T)\leq O(\gamma^2 2^{4\gamma}) f(S)$.
\end{reptheorem}
\begin{proof}
	Since $f$ is monotone, we assume that $|S|=|T|=r$. By Lemma~\ref{lem:exchangematroid}, there is a bijective mapping $g:S\setminus T\rightarrow T\setminus S$ such that $S-i+g(i)\in\mathcal{I}$ where $i\in S\setminus T$. 
	Since $S$ is a $(1+\frac{\epsilon}{n^2})$-approximate local optimum, for all $i \in S \setminus T$ we have
	$
	(1+\frac{\epsilon}{n^2})f(S)\geq f(S-i+g(i)).
	$
	That is,
	$
	\frac{\epsilon}{n^2}f(S)+B_i(S-i)\geq B_{g(i)}(S-i).
	$
	Using this we get
	\begin{align*}
	B_{g(i)}(S) & =B_{g(i)}(S-i)+A_{ig(i)}(S-i)\leq B_{g(i)}(S-i) +\gamma(\frac{B_{g(i)}(S-i)+B_{i}(S-i)}{r-1}) \\ & \leq \frac{2\gamma + r-1}{r -1} B_{i}(S-i) + \frac{\epsilon(\gamma+r-1)}{(r-1)n^2} f(S),
	\end{align*}
	where the equality follows from Lemma~\ref{lemma:discreteIntegral} and the first inequality from $\gamma$-meta-submodularity. Therefore,
	\begin{align*}
	\sum_{i\in S\setminus T} B_{g(i)}(S)\leq \frac{2\gamma + r-1}{r -1} \sum_{i\in S\setminus T} B_{i}(S-i) + o(1)f(S).
	\end{align*}
	Now, by Taylor's Theorem, Lemma~\ref{lemma:grad-ratio-ews}, and the above inequality, we have
	\begin{align*}
	f(S\cup T) & = F(\mathbbm{1}_S\vee \mathbbm{1}_T) = F(\mathbbm{1}_S+\mathbbm{1}_{T\setminus S}) = F(\mathbbm{1}_S)+\mathbbm{1}_{T\setminus S}^{\top}\nabla F(\mathbbm{1}_S+\epsilon' \mathbbm{1}_{T\setminus S}) \\ & \leq F(\mathbbm{1}_S)+2^{4\gamma} \mathbbm{1}_{T\setminus S}^{\top}\nabla F(\mathbbm{1}_S) = F(\mathbbm{1}_S)+2^{4\gamma}\sum_{i\in S\setminus T} B_{g(i)}(S) \\ & \leq (1+2^{4 \gamma} \cdot o(1))f(S)+\frac{2\gamma+r-1}{r-1}2^{4\gamma}\sum_{i\in S\setminus T} B_i(S-i)
	\end{align*}
	Therefore, using the monotonicity of $f$ and Lemma~\ref{lemma:sumGradEws} we get
	\begin{align*}
	f(T)\leq f(S\cup T)\leq 
	\Big[\frac{2\gamma+r-1}{r-1}2^{4\gamma}(5\gamma+2) + 1 +2^{4 \gamma} \cdot o(1) \Big] f(S).
	\end{align*}
\end{proof}

As discussed, 
one can get improved approximation factors by requiring additional conditions on the marginal gains of the set function $f$.  
We discuss this in the next section.

\section{Meta-Submodularity with Additional Second Order Conditions} 
\label{sec:second-order}

In this section we show that the modified local search algorithm can be used to find an $O(\gamma^{2})$-approximation for maximizing a  second-order submodular $\gamma$-MS function subject to a matroid constraint. Moreover if the function is supermodular, we improve the approximation to $O(\gamma^{3/2})$. Our result relies on the following key lemma, which  bounds the Taylor series expansion of the multilinear extension of second-order submodular functions.

\begin{lemma}
	Let $f:2^n\rightarrow \mathbb{R}$ be a non-negative, second-order submodular set function and $F$ be its multilinear extension. Then for any $R\subseteq [n]$, $\sum_{i\in R} B_i(R)\leq 2f(R)$. If $f$ is also monotone then for $x\in [0,1]^n$, $x^{\top} \nabla^2 F(x) x \leq 2F(x)$.
	\label{lemma:2ndorder-structure}
\end{lemma}
\begin{proof}
	For the first part, without loss of generality let $R=[r]$ (we can always relabel the elements so that this is true) and $R_i=[i]$. By Lemma~\ref{lemma:discreteIntegral}, we have
	\[
	\sum_{i\in R} B_i(R) = \sum_{ i = 1 }^r \big(f(\{ i\})+\sum_{j = 1}^r A_{ij} (R_{j-1})\big).
	\]
	Since $B_i(R_{i}) = B_i(R_{i-1})$, and $f(R_0)=f(\emptyset)=0$ we have
	\[
	2f(R) = 2\sum_{i = 1}^r B_i(R_{i}) = 2\sum_{i = 1}^r \big( f(\{i\}) + \sum_{j = 1}^i A_{ij} (R_{j-1})\big).
	\]
	Moreover, note that 
	\[
	\sum_{i=1}^r \sum_{j=1}^k A_{ij}(R_{j-1}) \leq 2 \sum_{i=1}^r \sum_{j=1}^i A_{ij}(R_{j-1})
	\]
	since
	\begin{align*}
	\sum_{i=1}^r \sum_{j=i+1}^r A_{ij}(R_{j-1}) &
	= \sum_{j=1}^r \sum_{i=1}^{j-1} A_{ij}(R_{j-1}) 
	=  \sum_{j=1}^r \sum_{i=1}^{j-1} A_{ji}(R_{j-1})
	\leq \sum_{j=1}^r \sum_{i=1}^{j-1} A_{ji}(R_{i-1}) 
	 \\ & = \sum_{j=1}^r \sum_{i=1}^{j} A_{ji}(R_{i-1})
	= \sum_{i=1}^r \sum_{j=1}^{i} A_{ij}(R_{j-1}),
	\end{align*}
	where the second equality follows from the fact that $A_{ij}(S) = A_{ji}(S)$ for all $i,j\in [n]$ and $S\subseteq [n]$, and the third equality from the fact that $A_{ii}(S)=0$ for all $i \in [n]$ and $S\subseteq [n]$. The inequality follows since $f$ is second-order submodular and $R_{j-1} \supseteq R_{i-1}$ if $j\geq i$.
	
	By non-negativity we also have that $2f(\{i\})\geq f(\{i\})$. This yields the first part of the lemma.
	
	We now discuss the second part. By the Taylor's Theorem, non-negativity, monotononicity and second-order submodularity, we have
	\[
	F(x) = F(0) + x^{\top} \nabla F(0) + \frac{1}{2} x^{\top} \nabla^2 F(\epsilon x) x \geq \frac{1}{2} x^{\top} \nabla^2 F(\epsilon x) x \geq \frac{1}{2} x^{\top} \nabla^2 F( x) x.
	\]
\end{proof}

Now, we are equipped to improve the approximation factor for meta-submodular functions with additional assumptions.

\begin{reptheorem}{thm:LS-2nd-order}
	Let $f\in \metas{\gamma}$ be second-order submodular
	(that is, $B_i$'s are submodular). 
	Let $\M=([n],\mathcal{I})$ be a matroid of rank $r$ and minimum circuit size of $c>2$. Let $T\in \mathcal{I}$ be an optimum set, i.e.,
	$
	T \in \argmax_{R\in \mathcal{I}} f(R),
	$
	and $S\in \mathcal{I}$ be an $(1+\frac{\epsilon}{n^2})$-approximate local optimum, i.e., for any $i$ and $j$ such that $S-i+j\in \mathcal{I}$,
	$
	(1+\frac{\epsilon}{n^2})f(S)\geq f(S-i+j),
	$
	where $\epsilon>0$ is a constant. Then $f(T)\leq O(\gamma+\frac{\gamma^2}{r}) f(S)$. So Algorithm~\ref{alg:localsearch} gives an $O(\gamma+\frac{\gamma^2}{r})$-approximation. If $f$ is also supermodular then Algorithm~\ref{alg:localsearch} gives an $O(\min\{\gamma + \frac{\gamma^2}{r}, \frac{\gamma r}{c-1} \}) \leq O(\gamma^{3/2})$-approximation.
\end{reptheorem}
\begin{proof}
	Since $f$ is monotone, we assume that $|S|=|T|=r$. By Lemma~\ref{lem:exchangematroid}, there is a bijective mapping $g:S\setminus T\rightarrow T\setminus S$ such that $S-i+g(i)\in\mathcal{I}$ where $i\in S\setminus T$.	Since $S$ is a $(1+\frac{\epsilon}{n^2})$-approximate local optimum, for all $i \in S \setminus T$ we have
	$
	(1+\frac{\epsilon}{n^2})f(S)\geq f(S-i+g(i)).
	$
	That is,
	\begin{equation}
	\label{eq:aux1}
	\frac{\epsilon}{n^2}f(S)+B_i(S-i)\geq B_{g(i)}(S-i).
	\end{equation}
	Using this we get
	\begin{align*}
	& B_{g(i)}(S) =B_{g(i)}(S-i)+A_{ig(i)}(S-i)\leq B_{g(i)}(S-i) +\gamma(\frac{B_{g(i)}(S-i)+B_{i}(S-i)}{r-1}) \\ & \leq \frac{2\gamma + r-1}{r -1} B_{i}(S-i) + \frac{\epsilon(\gamma+r-1)}{(r-1)n^2} f(S)
	= \Big(\frac{2\gamma}{r-1} + 1 \Big) B_{i}(S) + \frac{\epsilon(\gamma+r-1)}{(r-1)n^2} f(S),
	\end{align*}
	where the first equality follows from Lemma~\ref{lemma:discreteIntegral},  the first inequality from $\gamma$-meta-submodularity, and the last equality from $B_i(S) = B_i(S-i)$ for all $i \in [n]$ and $S \subseteq [n]$. Thus,
	\begin{align*}
	\sum_{i\in S\setminus T} B_{g(i)}(S) & 
	\leq \Big(\frac{2\gamma}{r-1} + 1 \Big) \sum_{i\in S\setminus T} B_{i}(S) + |S \setminus T| \cdot \frac{\epsilon(\gamma+r-1)}{(r-1)n^2}  f(S) \\ &
	\leq \Big(\frac{2\gamma}{r-1} + 1 \Big) \sum_{i\in S} B_{i}(S) + \frac{\epsilon(\gamma+r-1)}{(r-1)n} f(S)  \\ &
	\leq \Big(\frac{4\gamma}{r-1} + 2 + o(1) \Big) \cdot f(S).
	\end{align*}
	where the second inequality follows from monotonicity (i.e. $B_i(S) \geq 0$), and the last one follows from Lemma~\ref{lemma:2ndorder-structure}.

	Now, by Taylor's Theorem and the submodularity of the marginal gains of $f$ (i.e. the submodularity of $B_i$'s), $\gamma$-meta submodularity, and the above inequality, we have
	\begin{align*}
	f(T) & \leq f(S\cup T)  =  F(\mathbbm{1}_S+\mathbbm{1}_{T\setminus S}) \leq F(\mathbbm{1}_S)+\mathbbm{1}_{T\setminus S}^{\top}\nabla F(\mathbbm{1}_S) + \frac{1}{2} \mathbbm{1}_{T\setminus S}^{\top}\nabla^2 F(\mathbbm{1}_S) \mathbbm{1}_{T\setminus S} \\ & \leq F(\mathbbm{1}_S)+\Big( 1 + \frac{ \gamma |T \setminus S|}{|S|} \Big) \mathbbm{1}_{T\setminus S}^{\top}\nabla F(\mathbbm{1}_S) \leq F(\mathbbm{1}_S)+( 1 + \gamma ) \mathbbm{1}_{T\setminus S}^{\top}\nabla F(\mathbbm{1}_S) \\
	& = F(\mathbbm{1}_S)+ ( 1 +  \gamma )\sum_{i\in S\setminus T} B_{g(i)}(S) 
	\leq \Big( \frac{4 \gamma^2}{r-1} + \gamma \big( \frac{4}{r-1} + 2 + o(1) \big) + 3 + o(1) \Big) f(S) \\ & 
	= O \Big(\frac{\gamma^2}{r} + \gamma \Big) f(S).
	\end{align*}
	Now, we assume that $f$ is also supermodular. Let $M$ be the maximum weighted matching defined in line $7$ of Algorithm~\ref{alg:localsearch} and $S'$ be the node set of $M$. Let $S\cap S' = \{a_1,\ldots,a_p \}$ and $S' \setminus S = \{b_1,\ldots,b_p \}$ where $\{a_i,b_i\}$'s are the edges of $M$. Also, let $U_i = \{a_1,\ldots,a_i \}$ and $R_i = \{b_1,\ldots,b_i \}$, where $U_0=R_0=\emptyset$. Then since $M$ is a maximum weighted matching, we have
	\begin{align}
	\label{eq:localsecond1}
	\sum_{i\in S\setminus T} A_{i g(i)} (S) \leq  \frac{|S \setminus T|}{\floor{\frac{c-1}{2}}} \sum_{i=1}^p A_{a_i b_i} (S) \leq \frac{3r}{c-1} \sum_{i=1}^p A_{a_i b_i} (S),
	\end{align}
	where the second inequality follows from $\frac{c-1}{3}\leq \floor{\frac{c-1}{2}}$ (when $c>2$) and the assumption that $f$ is supermodular, which implies $A_{a_ib_i}$'s are non-negative.
	We also have that
	\begin{align}
	\label{eq:localsecond2}
	f(S') & = \sum_{i = 1}^p (f(U_i\cup R_i)-f(U_{i-1}\cup R_{i-1})) = \sum_{i=1}^p (B_{a_i}(U_{i-1}\cup R_{i-1})+B_{b_i}(U_{i-1}\cup R_{i-1}+a_i)) \nonumber \\ 
	& = \sum_{i=1}^p \Big(B_{a_i}(U_{i-1}\cup R_{i-1})+f(\{b_i\})+\sum_{j=1}^i A_{b_i a_j} (U_{j-1}) + \sum_{j=1}^{i-1} A_{b_i b_j} (U_{i-1} + a_i \cup R_{j-1} )\Big)
	\nonumber \\ 
	& = \sum_{i=1}^p \Big(B_{a_i}(U_{i-1}\cup R_{i-1})+A_{b_i a_i} (U_{i-1})+f(\{b_i\})+\sum_{j=1}^{i-1} A_{b_i a_j} (U_{j-1})
	\nonumber \\ & + \sum_{j=1}^{i-1} A_{b_i b_j} (U_{i-1}\cup R_{j-1} + a_i)\Big)
	\geq \sum_{i=1}^p A_{a_i b_i} (U_{i-1}) \geq \sum_{i=1}^p A_{a_i b_i} (S).
	\end{align}
	where the third equality 
	follows from Lemma~\ref{lemma:discreteIntegral}, the first inequality from monotonocity and supermodularity (i.e. all the $B_i$ and $A_{ij}$ terms are non-negative), and the last inequality from second-order submodularity and the fact that $U_i\subseteq S$ for any $i=1,\ldots,p$.

	Hence, by combining (\ref{eq:localsecond1}) and (\ref{eq:localsecond2}), we get 
	\begin{align}
	\label{eq:localsecond3}
	\sum_{i\in S\setminus T} A_{i g(i)} (S-i)= \sum_{i\in S\setminus T} A_{i g(i)} (S) \leq \frac{3r}{c-1} \sum_{i=1}^p A_{a_i b_i} (S) \leq \frac{3r}{c-1} f(S').
	\end{align}
	We have
	\begin{align*}
	f(T) & \leq f(S\cup T)  =  F(\mathbbm{1}_S+\mathbbm{1}_{T\setminus S}) \leq F(\mathbbm{1}_S)+\mathbbm{1}_{T\setminus S}^{\top}\nabla F(\mathbbm{1}_S) + \frac{1}{2} \mathbbm{1}_{T\setminus S}^{\top}\nabla^2 F(\mathbbm{1}_S) \mathbbm{1}_{T\setminus S} \\ & \leq F(\mathbbm{1}_S)+\Big( 1 + \frac{ \gamma |T\setminus S|}{|S|} \Big) \mathbbm{1}_{T\setminus S}^{\top}\nabla F(\mathbbm{1}_S) \leq F(\mathbbm{1}_S)+( 1 + \gamma ) \mathbbm{1}_{T\setminus S}^{\top}\nabla F(\mathbbm{1}_S) \\
	& = F(\mathbbm{1}_S)+ ( 1 +  \gamma )\sum_{i\in S\setminus T} B_{g(i)}(S) \\ & = f(S) + (1+\gamma)(\sum_{i\in S\setminus T} B_{g(i)}(S-i)+\sum_{i\in S\setminus T} A_{ig(i)}(S-i)) \\ & \leq
	f(S)+(1+\gamma)\Big(\frac{r\epsilon}{n^2}f(S)+\sum_{i\in S\setminus T}B_i(S-i)+\frac{3r}{c-1} f(S') \Big) \\ & \leq f(S)+(1+\gamma)\Big(\frac{r\epsilon}{n^2}f(S)+2f(S)+\frac{3r}{c-1} f(S')\Big) = O\big(\frac{\gamma r}{c-1} \big) \max \{f(S),f(S')\}.
	\end{align*}
	where the second inequality follows from Taylor's Theorem and second-order submodularity (i.e. the non-positivity of the third order derivatives), the third inequality from $\gamma$-meta submodularity, the fifth inequality from (\ref{eq:aux1}) and (\ref{eq:localsecond3}), and the second to last inequality from Lemma~\ref{lemma:2ndorder-structure}. 
	We then have that if $r\leq \sqrt{\gamma}$ then $\gamma r = O(\gamma^{3/2})$, and if $r \geq \sqrt{\gamma}$ then $\frac{\gamma^2}{r}+\gamma= O(\gamma^{3/2})$. Therefore, $f(T)\leq O(\gamma^{3/2})\max \{ f(S), f(S')\}$.
\end{proof}

\section{Conclusions}
Maximizing a set function  subject to cardinality (or matroid) constraint can capture problems with sweeping applications. The setting is too general, however, to allow   algorithms with good performance on all data sets.  It remains an interesting direction to classify those set functions which lead to tractable formulations.
This is the key question considered in this work. We provide a ``spectrum
of tractability'' by defining a new meta-submodularity parameter $\gamma$ associated with any monotone set function.
These families capture for low values of $\gamma$ several widely known tractable classes, such
as submodular functions ($\gamma=0$) or metric diversity ($\gamma=1$).
We then show that 
there exist efficient (in theory and practice) algorithms which have maximization approximation guarantees which are function of $\gamma$  alone. 

\bibliographystyle{plain}
\bibliography{references_diversity}
\newpage
\appendix
\section{Appendix: Preliminaries} 
\label{app:preliminaries}

The following result describes the connection between the terms $A_{ij}$ and $B_i$. One can see it as a discrete integral formula.

\begin{replemma}{lemma:discreteIntegral}
Let $f:2^{[n]} \to \R$, $i\in [n]$, and $R=\{v_1,\ldots, v_r \}\subseteq [n]$. Moreover, let $R_m=\{v_1,\ldots, v_m \}$ for $1\leq m\leq r$ and $R_0=\emptyset$. Then
\[
B_i(R) = f(\{i\}) + \sum_{j=1}^{r} A_{i v_j} (R_{j-1}).
\]
\end{replemma}
\begin{proof}
First, we consider the case where $i \notin R$. Then $B_i(R) = f(R+i)-f(R)$ and the right hand side is equal to
\begin{align*}
& f(R_{r-1} + i + v_r) - f(R_{r-1} - i + v_r) - f(R_{r-1} + i - v_r) + f(R_{r-1} - i - v_r)
\\ & + f(R_{r-2} + i + v_{r-1}) - f(R_{r-2} - i + v_{r-1}) - f(R_{r-2} + i - v_{r-1}) + f(R_{r-2} - i - v_{r-1})
\\ & + \cdots
\\ & + f(R_{1} + i + v_{2}) - f(R_{1} - i + v_{2}) - f(R_{1} + i - v_{2}) + f(R_{1} - i - v_{2})
\\ & + f(R_{0} + i + v_{1}) - f(R_{0} - i + v_{1}) - f(R_{0} + i - v_{1}) + f(R_{0} - i - v_{1})
\\ & + f(\{i\})
\\ & = f(R + i) - f(R) - f(R_{r-1} + i) + f(R_{r-1})
\\ & + f(R_{r-1} + i) - f(R_{r-1}) - f(R_{r-2} + i) + f(R_{r-2})
\\ & + \cdots
\\ & + f(R_{2} + i) - f(R_{2}) - f(R_{1} + i) + f(R_{1})
\\ & + f(R_{1} + i) - f(R_{1}) - f(R_{0} + i) + f(R_{0})
\\ & + f(\{i\})
\\ & = f(R+i) - f(R)
\end{align*}
The last equality holds because the third and the fourth elements of each line cancel out the first and the second element of the next line (except for the last two lines), respectively. For the last two lines, note that $f(R_0)=f(\emptyset) = 0$ and $f(R_0+i)=f(\{i \})$.

Now, we consider the case that $i \in R$. Let $i=v_j$. Then $B_i(R)=f(R)-f(R-i)$ and the right hand side is equal to
\begin{align*}
& f(R_{r-1} + i + v_r) - f(R_{r-1} - i + v_r) - f(R_{r-1} + i - v_r) + f(R_{r-1} - i - v_r)
\\ & + f(R_{r-2} + i + v_{r-1}) - f(R_{r-2} - i + v_{r-1}) - f(R_{r-2} + i - v_{r-1}) + f(R_{r-2} - i - v_{r-1})
\\ & + \cdots
\\ & + f(R_{j} + i + v_{j+1}) - f(R_{j} - i + v_{j+1}) - f(R_{j} + i - v_{j+1}) + f(R_{j} - i - v_{j+1})
\\ & + f(R_{j-1} + i + v_{j}) - f(R_{j-1} - i + v_j) - f(R_{j-1} + i - v_j) + f(R_{j-1} - i - v_j)
\\ & + f(R_{j-2} + i + v_{j-1}) - f(R_{j-2} - i + v_{j-1}) - f(R_{j-2} + i - v_{j-1}) + f(R_{j-2} - i - v_{j-1})
\\ & + \cdots
\\ & + f(R_{1} + i + v_{2}) - f(R_{1} - i + v_{2}) - f(R_{1} + i - v_{2}) + f(R_{1} - i - v_{2})
\\ & + f(R_{0} + i + v_{1}) - f(R_{0} - i + v_{1}) - f(R_{0} + i - v_{1}) + f(R_{0} - i - v_{1})
\\ & + f(\{i\})
\\ & = f(R) - f(R-i) - f(R_{r-1}) + f(R_{r-1} - i)
\\ & + f(R_{r-1}) - f(R_{r-1} - i) - f(R_{r-2}) + f(R_{r-2} - i)
\\ & + \cdots
\\ & + f(R_{j+1}) - f(R_{j+1} - i) - f(R_{j}) + f(R_{j-1})
\\ & + f(R_{j}) - f(R_{j}) - f(R_{j-1}) + f(R_{j-1})
\\ & + f(R_{j}) - f(R_{j-1}) - f(R_{j-2} + i) + f(R_{j-2})
\\ & + \cdots
\\ & + f(R_{2} + i) - f(R_{2}) - f(R_{1} + i) + f(R_{1})
\\ & + f(R_{1} + i) - f(R_{1}) - f(R_{0} + i) + f(R_{0})
\\ & + f(\{i\})
\\ & = f(R) - f(R-i).
\end{align*}
Like before the last equality holds because the last two terms of each line cancels out the first two terms of the next line except for the last two lines, the first $f(R_j)$ line and the $f(R_{j+1})$ line. The terms of the first $f(R_j)$ line cancel each other out, while the last two terms of the $f(R_{j+1})$ line cancel the first two terms of the second $f(R_{j})$ line.
\end{proof}

The following result connects the first and second order marginal gains $B_i$ and $A_{ij}$, to the first and second order partial derivatives of the multilinear extension.

\begin{lemma}[\cite{vondrak2008optimal}]
\label{lemma:gradientbi}
Let $f$ be a set function and $F$ its multilinear function. Then for any $x=(x_1,\ldots,x_n)\in [0,1]^n$ and $i,j \in [n]$,
\begin{align*}
\nabla_i F(x) & =
\E_{R \sim x}[B_i(R)] =
\sum_{R \subseteq [n]} B_i(R)p_x(R) \\ & =
\sum_{R\subseteq [n] - i} [f(R+i)-f(R)] \prod_{v\in R} x_v \prod_{v\in [n]\setminus(R+i)} (1-x_v),
\end{align*}
and, 
\begin{align*}
& \nabla_{ij}^2 F(x) = 
\E_{R \sim x}[A_{ij}(R)] =
\sum_{R \subseteq [n]} A_{ij}(R)p_x(R) \\ & =
\sum_{R\subseteq [n]-i-j} [f(R+i+j)-f(R+i)-f(R+j)+f(R)] \prod_{v\in R} x_v \prod_{v\in [n]\setminus(R+i+j)} (1-x_v).
\end{align*}
\end{lemma}
\begin{proof}
First of all, note that if $i\notin R$ then $B_i(R+i)=B_i(R)$. Now, we write the multilinear function
\begin{align*}
F(x) & = \sum_{R\subseteq [n]} f(R) \prod_{v\in R} x_v \prod_{v\in [n]\setminus R} (1-x_v) \\ & = \sum_{R\subseteq [n] - i} (f(R+i)x_i+f(R)(1-x_i)) \prod_{v\in R} x_v \prod_{v\in [n]\setminus (R+i)} (1-x_v).
\end{align*}
Therefore 
\begin{align*}
\nabla_i F(x) & = \sum_{R\subseteq [n] - i} (f(R+i)-f(R)) \prod_{v\in R} x_v \prod_{v\in [n]\setminus(R+i)} (1-x_v) \\ & = x_i \sum_{R\subseteq [n] - i} (f(R+i)-f(R)) \prod_{v\in R} x_v \prod_{v\in [n]\setminus(R+i)} (1-x_v) \\ & + (1-x_i) \sum_{R\subseteq [n] - i} (f(R+i)-f(R)) \prod_{v\in R} x_v \prod_{v\in [n]\setminus(R+i)} (1-x_v) \\ & = \sum_{R\subseteq [n] - i} (f(R+i)-f(R)) \prod_{v\in R+i} x_v \prod_{v\in [n]\setminus(R+i)} (1-x_v) \\ & + \sum_{R\subseteq [n] - i} (f(R+i)-f(R)) \prod_{v\in R} x_v \prod_{v\in [n]\setminus R} (1-x_v) \\ & = \sum_{R\subseteq [n] - i} B_i(R+i) p_x(R+i) + \sum_{R\subseteq [n] - i} B_i(R) p_x(R) \\ & = \sum_{R \subseteq [n]} B_i(R)p_x(R).
\end{align*}
Now, to prove the other part of the lemma, we write the multilinear function again.
\begin{align*}
F(x) &  = \sum_{R\subseteq [n]} f(R) \prod_{v\in R} x_v \prod_{v\in [n]\setminus R} (1-x_v) \\ & = x_i x_j \sum_{R\subseteq [n]-i-j} f(R+i+j) \prod_{v\in R} x_v \prod_{v\in [n]\setminus(R+i+j)} (1-x_v) \\ & + x_i (1-x_j) \sum_{R\subseteq [n]-i-j} f(R+i) \prod_{v\in R} x_v \prod_{v\in [n]\setminus(R+i+j)} (1-x_v) \\ & + (1-x_i) x_j \sum_{R\subseteq [n]-i-j} f(R+j) \prod_{v\in R} x_v \prod_{v\in [n]\setminus(R+i+j)} (1-x_v) \\ & + (1-x_i) (1-x_j) \sum_{R\subseteq [n]-i-j} f(R) \prod_{v\in R} x_v \prod_{v\in [n]\setminus(R+i+j)} (1-x_v).
\end{align*}
Therefore, by using the fact that $x_i x_j + (1-x_i) x_j + x_i (1-x_j) + (1-x_i)(1-x_j) = 1$, and $A_{ij}(R+i+j)=A_{ij}(R+i)=A_{ij}(R+j)=A_{ij}(R) = f(R+i+j) - f(R+i) - f(R+j) + f(R)$ for $R\subseteq [n]-i-j$, we have
\begin{align*}
\nabla_{ij}^2 F(x) & = \sum_{R\subseteq [n]-i-j} (f(R+i+j)-f(R+i)-f(R+j)+f(R)) \prod_{v\in R} x_v \prod_{v\in [n]\setminus(R+i+j)} (1-x_v)
\\ & = x_i x_j \sum_{R\subseteq [n]-i-j} A_{ij}(R+i+j) \prod_{v\in R} x_v \prod_{v\in [n]\setminus(R+i+j)} (1-x_v)
\\ & + (1-x_i) x_j \sum_{R\subseteq [n]-i-j} A_{ij}(R+j) \prod_{v\in R} x_v \prod_{v\in [n]\setminus(R+i+j)} (1-x_v)
\\ & + x_i (1-x_j) \sum_{R\subseteq [n]-i-j} A_{ij}(R+i) \prod_{v\in R} x_v \prod_{v\in [n]\setminus(R+i+j)} (1-x_v)
\\ & + (1-x_i) (1-x_j) \sum_{R\subseteq [n]-i-j} A_{ij}(R) \prod_{v\in R} x_v \prod_{v\in [n]\setminus(R+i+j)} (1-x_v)
\\ & = \sum_{R\subseteq [n]-i-j} A_{ij}(R+i+j) \prod_{v\in R+i+j} x_v \prod_{v\in [n]\setminus R} (1-x_v)
\\ & + \sum_{R\subseteq [n]-i-j} A_{ij}(R+j) \prod_{v\in R+j} x_v \prod_{v\in [n]\setminus(R+i)} (1-x_v)
\\ & + \sum_{R\subseteq [n]-i-j} A_{ij}(R+i) \prod_{v\in R+i} x_v \prod_{v\in V\setminus(R+i)} (1-x_v)
\\ & + \sum_{R\subseteq [n]-i-j} A_{ij}(R) \prod_{v\in R} x_v \prod_{v\in [n]\setminus R} (1-x_v)
\\ & = \sum_{R\subseteq [n]-i-j} A_{ij}(R+i+j) p_x(R+i+j)
\\ & + \sum_{R\subseteq [n]-i-j} A_{ij}(R+j) p_x(R+j)
\\ & + \sum_{R\subseteq [n]-i-j} A_{ij}(R+i) p_x(R+i)
\\ & + \sum_{R\subseteq [n]-i-j} A_{ij}(R) p_x(R)
\\ & = \sum_{R\subseteq [n]} A_{ij}(R) p_x(R).
\end{align*}
\end{proof}

\section{Appendix: Meta-Submodular Family}
\label{app:metasubfamily}

In this section, we discuss the meta-submodularity parameter of the class of meta-submodular functions (defined by Kleinberg et al.~\cite{KleinbergPR98}) and the class of proportionally submodular functions (defined by Borodin et al.~\cite{borodin2015proportionally}).

\begin{proposition}
\label{prop:0metakleinberg}
$f$ is $0$-meta-submodular if and only if it is meta-submodular (by Kleinberg et al. definition~\cite{KleinbergPR98}).
\end{proposition}
\begin{proof}
Kleinberg et al \cite{KleinbergPR98} show that a set function $f$ is meta-submodular if and only if
\begin{equation*}
f(S+i) - f(S) \geq f(T+i) - f(T), \;\;  \forall \emptyset \neq S \subseteq T, \; \forall i \notin T.
\end{equation*}
The above is clearly equivalent to
\begin{equation}
\label{eqn: Kleinberg-ms}
f(S+i) - f(S) \geq f(S+j+i) - f(S+j), \;\;  \forall S \neq \emptyset, \; \forall i \neq j \notin S.
\end{equation}
Then
\begin{align*}
    & f \mbox{ is 0-meta submodular}\\
    \iff & A_{ij}(S) \leq 0, \;\; \forall S \neq \emptyset, \; \forall i,j \in V\\
    \iff & f(S+i+j)-f(S+i)-f(S+j)+f(S) \leq 0, \;\;  \forall S \neq \emptyset, \; \forall i,j \in V \\
    \iff & f(S+i) - f(S) \geq f(S+j+i) - f(S+j), \;\;  \forall S \neq \emptyset, \; \forall i,j \in V \\
    \iff & f(S+i) - f(S) \geq f(S+j+i) - f(S+j), \;\;  \forall S \neq \emptyset, \; \forall i \neq j \notin S \\
    \iff & (\ref{eqn: Kleinberg-ms}) \mbox{ holds}.
\end{align*}

\end{proof}

\begin{proposition}
\label{prop:borodinweakmeta}
Any monotone propotionally submodular function is $1$-meta-submodular.
\end{proposition}
\begin{proof}
The proof is by case analysis.
\begin{itemize}
\item If $i,j\notin R$ then using the proportional submodularity property we have
\[
(|R|+2)f(R)+(|R|)f(R+i+j)\leq (|R|+1)f(R+i)+(|R|+1)f(R+j),
\]
which means
\[
|R| \cdot (f(R)+f(R+i+j)-f(R+i)-f(R+j))\leq f(R+i)+f(R+j)-2f(R).
\]
Hence
\begin{align*}
& f(R+i+j)-f(R+i-j)-f(R+j-i)+f(R-i-j) \\ & = f(R+i+j)-f(R+i)-f(R+j)+f(R)\\ & \leq \frac{f(R+i)-f(R)+f(R+j)-f(R)}{|R|}\\ & = \frac{f(R+i)-f(R-i)+f(R+j)-f(R-j)}{|R|}.
\end{align*}
\item If $i,j\in R$ then by proportional submodularity we have
\[(|R|-2)f(R)+(|R|)f(R-i-j)\leq (|R|-1)f(R-i)+(|R|-1)f(R-j),
\]
which means
\[
|R| \cdot (f(R)+f(R-i-j)-f(R-i)-f(R-j))\leq 2f(R)-f(R-i)-f(R-j).
\]
Hence
\begin{align*}
& f(R+i+j)-f(R+i-j)-f(R+j-i)+f(R-i-j) \\ & = f(R)-f(R-j)-f(R-i)+f(R-i-j)\\ & \leq \frac{f(R)-f(R-i)+f(R)-f(R-j)}{|R|} \\ & = \frac{f(R+i)-f(R-i)+f(R+j)-f(R-j)}{|R|}.
\end{align*}
\item If $i\in R$ and $j\notin R$ then using the proportional submodularity property we have
\[(|R|-1)f(R+j)+(|R|+1)f(R-i)\leq (|R|)f(R)+(|R|)f(R+j-i),
\]
which means
\begin{align*}
& |R| \cdot (f(R+j)+f(R-i)-f(R)-f(R+j-i))\leq f(R+j)-f(R-i) \\ & =  f(R+j) - f(R-j) + f(R+i) - f(R-i),
\end{align*}
where the equality is correct because $f(R)=f(R-j)=f(R+i)$.
Hence
\begin{align*}
& f(R+i+j)-f(R+i-j)-f(R+j-i)+f(R-i-j) \\ & = f(R+j)-f(R)-f(R+j-i)+f(R-i)\\ & \leq \frac{f(R+j)-f(R-i)}{|R|}\\ & = \frac{f(R+i)-f(R-i)+f(R+j)-f(R-j)}{|R|}.
\end{align*}
\end{itemize}
\end{proof}

\begin{proposition}
\label{prop:semimetricmeta}
Let $g(R):=\sum_{q\in R} g(q)$ be a non-negative modular function and $d(R)=\sum_{\{q,q'\}\subseteq R} A(q,q')$ be a diversity function such that $A$ is a $\gamma$-semi-metric distance and $\gamma\geq 1$. Then $f(R) := d(R) + g(R)$ is a $\gamma$-MS function.
\end{proposition}
\begin{proof}
We have $f(R)=\sum_{q\in R} g(q) + \sum_{\{q,q'\}\subseteq R} A(q,q')$. The proof goes by case analysis as follows.
\begin{itemize}
\item If $i,j\notin R$, we have
\begin{align*}
& |R|A_{ij}(R) = |R|(f(R+i+j)-f(R+i-j)-f(R-i+j)+f(R-i-j)) \\ & = |R|(\sum_{q\in R+i+j} g(q) + \sum_{\{q,q'\}\subseteq R+i+j} A(q,q') - \sum_{q\in R+i} g(q) - \sum_{\{q,q'\}\subseteq R+i} A(q,q') \\ & - \sum_{q\in R+j} g(q) - \sum_{\{q,q'\}\subseteq R+j} A(q,q') + \sum_{q\in R} g(q) + \sum_{\{q,q'\}\subseteq R} A(q,q')) \\ & = |R| A(i,j).
\end{align*}
We also have
\begin{align*}
\gamma(B_i(R)+B_j(R)) & = \gamma(f(R+i)-f(R-i)+f(R+j)-f(R-i)) \\ & = \gamma(\sum_{q\in R+i} g(q) + \sum_{\{q,q'\}\subseteq R+i} A(q,q') - \sum_{q\in R} g(q) - \sum_{\{q,q'\}\subseteq R} A(q,q') \\ & + \sum_{q\in R+j} g(q) + \sum_{\{q,q'\}\subseteq R+j} A(q,q') - \sum_{q\in R} g(q) - \sum_{\{q,q'\}\subseteq R} A(q,q')) \\ & = \gamma g(i)+\gamma g(j)+\gamma\sum_{q\in R} A(i,q) + \gamma\sum_{q\in R} A(j,q).
\end{align*}
Therefore $|R|A_{ij}(R)\leq \gamma(B_i(R)+B_j(R))$ because $g$ is non-negative and $A$ is a $\gamma$-semi-metric distance.
\item If $i,j\in R$, we have
\begin{align*}
|R|A_{ij}(R) & = |R|(f(R+i+j)-f(R+i-j)-f(R-i+j)+f(R-i-j)) \\ & = |R|(\sum_{q\in R} g(q) + \sum_{\{q,q'\}\subseteq R} A(q,q') - \sum_{q\in R-j} g(q) - \sum_{\{q,q'\}\subseteq R-j} A(q,q') \\ & - \sum_{q\in R-i} g(q) - \sum_{\{q,q'\}\subseteq R-i} A(q,q') + \sum_{q\in R-i-j} g(q) + \sum_{\{q,q'\}\subseteq R-i-j} A(q,q')) \\ & = |R| A(i,j).
\end{align*}
We also have
\begin{align*}
\gamma(B_i(R)+B_j(R)) & = \gamma(f(R+i)-f(R-i)+f(R+j)-f(R-i)) \\ & = \gamma(\sum_{q\in R} g(q) + \sum_{\{q,q'\}\subseteq R} A(q,q') - \sum_{q\in R-i} g(q) - \sum_{\{q,q'\}\subseteq R-i} A(q,q') \\ & + \sum_{q\in R} g(q) + \sum_{\{q,q'\}\subseteq R} A(q,q') - \sum_{q\in R-j} g(q) - \sum_{\{q,q'\}\subseteq R-j} A(q,q')) \\ & = \gamma g(i)+\gamma g(j)+2\gamma A(i,j) +\gamma\sum_{q\in R-i-j} A(i,q) + \gamma\sum_{q\in R-i-j} A(j,q).
\end{align*}
Therefore $|R|A_{ij}(R)\leq \gamma(B_i(R)+B_j(R))$ because $g$ is non-negative, $A$ is a $\gamma$-semi-metric distance, and $\gamma\geq 1$.
\item If $i\in R$ and $j\notin R$, we have
\begin{align*}
|R|A_{ij}(R) & = |R|(f(R+i+j)-f(R+i-j)-f(R-i+j)+f(R-i-j)) \\ & = |R|(\sum_{q\in R+j} g(q) + \sum_{\{q,q'\}\subseteq R+j} A(q,q') - \sum_{q\in R} g(q) - \sum_{\{q,q'\}\subseteq R} A(q,q') \\ & - \sum_{q\in R-i+j} g(q) - \sum_{\{q,q'\}\subseteq R-i+j} A(q,q') + \sum_{q\in R-i} g(q) + \sum_{\{q,q'\}\subseteq R-i} A(q,q')) \\ & = |R| A(i,j).
\end{align*}
We also have
\begin{align*}
\gamma(B_i(R)+B_j(R)) & = \gamma(f(R+i)-f(R-i)+f(R+j)-f(R-i)) \\ & = \gamma(\sum_{q\in R} g(q) + \sum_{\{q,q'\}\subseteq R} A(q,q') - \sum_{q\in R-i} g(q) - \sum_{\{q,q'\}\subseteq R-i} A(q,q') \\ & + \sum_{q\in R+j} g(q) + \sum_{\{q,q'\}\subseteq R+j} A(q,q') - \sum_{q\in R} g(q) - \sum_{\{q,q'\}\subseteq R} A(q,q')) \\ & = \gamma g(i)+\gamma g(j)+\gamma A(i,j) +\gamma\sum_{q\in R-i} A(i,q) + \gamma\sum_{q\in R-i} A(j,q).
\end{align*}
Therefore $|R|A_{ij}(R)\leq \gamma(B_i(R)+B_j(R))$ because $g$ is non-negative, $A$ is a $\gamma$-semi-metric distance, and $\gamma\geq 1$.
\end{itemize}
\end{proof}

\section{Appendix: One-Sided Smoothness and Meta-Submodularity}
\label{app:smoothness}

In this section we discuss the connection between meta-submodularity of a function and the smoothness of its multilinear extension. We show that the smoothness of the multilinear extension results in the meta-submodularity of the underlying set function.

\begin{repproposition}{prop:smoothisinmeta}
Let $f$ be a set function and $F$ be its multilinear extension. If $F$ is one-sided $(\gamma/2)$-smooth, then $f$ is $\gamma$-meta-submodular.
\end{repproposition}
\begin{proof}
Let non-empty $R\subseteq [n]$ and $i,j\in [n]$. The inequality from one-sided $(\gamma/2)$-smoothness for
$u=\mathbbm{1}_{\{i,j\}}$ and $x=\mathbbm{1}_R$ yields:
\[
\frac{1}{2} (2 u_i u_j \nabla^2 F_{ij}(x)) \leq \frac{\gamma}{2} \frac{u_i+u_j}{||x||_1} (u_i \nabla_i F(x) + u_j \nabla_j F(x))
\]

\noindent
Since $u_i=u_j=1$, $||x||_1=|R|$, $\nabla^2 F_{ij}(x)=A_{ij}(R)$, and $\nabla_i F(x) + \nabla_j F(x)=B_i(R)+B_j(R)$ we obtain
the $\gamma$-meta-submodular inequality.
\end{proof}

\subsection{Smoothness of Supermodular $\gamma$-Meta-Submodular Functions}

In this section we show that the multilinear extension of a supermodular $\gamma$-meta-submodular function is one-sided $O(\gamma)$-smooth. We do this by proving the expectation inequality for these functions and using Lemma~\ref{lemma:probabilisticVersion}.

\begin{lemma}
\label{lemma:supmetasmooth}
Let $f:2^{[n]} \to \R_+$ be a non-negative, monotone, supermodular, $\gamma$-meta-submodular set function. Let $x\in [0,1]^n\setminus\{\vec{0}\}$ and $R\subseteq [n]$ such that $1 \leq |R| < ||x||_1$. Then for all $i,j \in [n]$ we have
\[
(||x||_1 - |R|)A_{ij}(R)p_x(R)\leq 2 \gamma \sum_{e\in [n]\setminus R} (\frac{B_i(R+e)+B_j(R+e)}{|R|+1})p_x(R+e).
\]
Also, for the empty set,
\[
(||x||_1)A_{ij}(\emptyset)p_x(\emptyset)\leq 2(\gamma+1)\sum_{e\in [n]} (B_i(\{e\})+B_j(\{e\}))p_x(\{e\}).
\]
\end{lemma}
\begin{proof}
Let $|R| = r$. 
Note that $r<n$ because $|R|=r<||x||_1$. Also, note that if $x_e=1$ for some $e\in [n]\setminus R$ then $p_x(R)=0$, which means that the left hand side is zero. In that case, the inequality holds because $f$ is monotone and the right hand side is non-negative. Hence, we assume that $x_e<1$ for all $e\in [n]\setminus R$. We know that 
\[
\sum_{e\in [n]} x_e = ||x||_1.
\]
Therefore, because each $x_e\leq 1$,
\[
\sum_{e\in [n]\setminus R} x_e = ||x||_1 - \sum_{e\in R} x_e \geq ||x||_1 - \sum_{e\in R} 1 = ||x||_1 - |R|.
\]
Hence, since $0< 1-x_e\leq 1$ for all $e \in [n]\setminus R$, we get
\begin{align*}
(||x||_1 - |R|)A_{ij}(R)p_x(R) & \leq \sum_{e\in [n]\setminus R} x_e A_{ij}(R)p_x(R)
\\ & \leq \sum_{e\in [n]\setminus R} \frac{x_e}{1-x_e} A_{ij}(R)p_x(R)
\\ & = \sum_{e\in [n]\setminus R} A_{ij}(R)p_x(R+e).
\end{align*}
\noindent
Moreover, $2|R|\geq |R|+1$ because $|R|\geq 1$, and we have
\[
\sum_{e\in [n]\setminus R} A_{ij}(R)p_x(R+e) \leq 2 \sum_{e\in [n]\setminus R} \frac{|R|A_{ij}(R)}{|R|+1}p_x(R+e).
\]
\noindent
Using the $\gamma$-meta-submodularity and supermodularity we have
\begin{align*}
2 \sum_{e\in [n]\setminus R} \frac{|R|A_{ij}(R)}{|R|+1}p_x(R+e) & \leq 
2 \gamma \sum_{e\in [n]\setminus R} \frac{B_i(R)+B_j(R)}{|R|+1}p_x(R+e)
\\ & \leq 2 \gamma \sum_{e\in [n]\setminus R} \frac{B_i(R+e)+B_j(R+e)}{|R|+1}p_x(R+e)
\end{align*}
\noindent
Combining all of these inequalities yields the first part of the lemma. 
\noindent
For the second part of the lemma, we consider the set $\{i,j,e\}$. By Lemma~\ref{lemma:discreteIntegral} and the $\gamma$-meta-submodularity, we have
\begin{align*}
f(\{i,j,e\}) & = B_i(\{j,e\})+B_j(\{e\})+f(\{e\})
\\ & = A_{ij}(\{e\})+ B_i(\{e\})+B_j(\{e\}) + f(\{e\})
\\ & \leq (\gamma+1) (B_i(\{e\})+B_j(\{e\})) + f(\{e\}).
\end{align*}
Also, by Lemma~\ref{lemma:discreteIntegral}, we have
\begin{align*}
f(\{i,j,e\}) & = B_i(\{j,e\})+B_j(\{e\})+f(\{e\})
\\ & = A_{ie}(\{j\})+ A_{ij}(\emptyset) + f(\{i\})+B_j(\{e\}) + f(\{e\}).
\end{align*}
Therefore
\[
A_{ie}(\{j\})+ A_{ij}(\emptyset) + f(\{i\})+B_j(\{e\}) + f(\{e\}) \leq (\gamma+1)(B_i(\{e\})+B_j(\{e\})) + f(\{e\}).
\]
Hence, because $f$ is non-negative, monotone and supermodular, it follows that
\begin{equation}
\label{ineq1}
A_{ij}(\emptyset)\leq A_{ie}(\{j\})+ A_{ij}(\emptyset) + f(\{i\})+B_j(\{e\}) \leq (\gamma+1)(B_i(\{e\})+B_j(\{e\})).
\end{equation}
\noindent
Moreover, because $f$ is non-negative and monotone, we have
\begin{align*}
A_{ij}(\emptyset) & =f(\{ i,j\})-f(\{i\})-f(\{j\})+f(\emptyset)= B_j(\{i\})-f(\{j\}) \\ & \leq B_j(\{i\}) + B_i(\{i\})\leq (\gamma+1)(B_j(\{i\}) + B_i(\{i\})),
\end{align*}
and
\begin{align*}
A_{ij}(\emptyset) & =f(\{ i,j\})-f(\{i\})-f(\{j\})+f(\emptyset)= B_i(\{j\})-f(\{i\}) \\ & \leq B_i(\{j\}) + B_j(\{j\})\leq (\gamma+1)(B_i(\{j\}) + B_j(\{j\})).
\end{align*}
If $x_e=1$ for an $e\in [n]$ then $p_x(\emptyset)=0$ and the inequality holds because the left hand side is zero and the right hand side is non-negative (since $f$ is monotone). Therefore, we assume that $x_e<1$ for all $e\in [n]$.
Combining the above inequalities, we have
\begin{align*}
(||x||_1)A_{ij}(\emptyset)p_x(\emptyset) & = \sum_{e\in [n]} x_e A_{ij}(\emptyset)p_x(\emptyset)
\\ & \leq \sum_{e\in [n]} \frac{x_e}{1-x_e} A_{ij}(\emptyset)p_x(\emptyset)
\\ & = \sum_{e\in [n]} A_{ij}(\emptyset)p_x(\{e\})
\\ & \leq (\gamma+1) \sum_{e\in [n]} (B_i(\{e\})+B_j(\{e\}))p_x(\{e\}),
\end{align*}
\noindent
where the last inequality follows from ($\ref{ineq1}$). This completes the proof.

\end{proof}

\begin{lemma}
\label{lemma:expIneq}
Let $f$ be a non-negative, monotone, supermodular, $\gamma$-meta-submodular set function and $F$ be its multilinear function. Then for any $x\in [0,1]^n\setminus\{\vec{0}\}$ and $i,j\in [n]$,
\[
||x||_1\nabla_{ij}^2 F(x)\leq (\max \{3\gamma,2\gamma+1 \})(\nabla_i F(x)+\nabla_j F(x)).
\]
\end{lemma}
\begin{proof}
By using Lemma~\ref{lemma:supmetasmooth} for all the sets of size less than $||x||_1$, we can write
\begin{align}
\label{4grad}
 &(||x||_1)A_{ij}(\emptyset)p_x(\emptyset) + \sum_{\substack{R\subseteq [n] \\ 1\leq |R| < ||x||_1}} (||x||_1 - |R|)A_{ij}(R)p_x(R) \nonumber
 \\ & \leq (\gamma+1)\sum_{e\in [n]} (B_i(\{e\})+B_j(\{e\}))p_x(\{e\}) \\ & + 2\gamma \sum_{\substack{R\subseteq [n] \\ 1\leq |R| < ||x||_1}} \sum_{e\in [n]\setminus R} (\frac{B_i(R+e)+B_j(R+e)}{|R|+1})p_x(R+e) \nonumber
 \\ & = (\gamma+1)\sum_{e\in [n]} (B_i(\{e\})+B_j(\{e\}))p_x(\{e\}) + 2\gamma\sum_{\substack{R\subseteq [n] \\ 2\leq |R| < ||x||_1+1}}  (B_i(R)+B_j(R))p_x(R) \nonumber
 \\ & \leq \max \{\gamma+1,2\gamma \}  \sum_{R\subseteq [n]} (B_i(R)+B_j(R))p_x(R) = \max \{\gamma+1,2\gamma \} (\nabla_i F(x)+\nabla_j F(x)),
\end{align}
where the equality follows from a simple counting argument, and in the last inequality we used the monotonicity of $f$ (i.e., the $B_i$'s are non-negative). 

\noindent
By $\gamma$-meta-submodularity, we also have that
\begin{align}
\label{ineq2}
 & \sum_{\substack{R\subseteq [n] \\ 1\leq |R| < ||x||_1}} |R|A_{ij}(R)p_x(R) + \sum_{\substack{R\subseteq [n] \\ |R| \geq ||x||_1}} (||x||_1) A_{ij}(R)p_x(R) \nonumber
 \\ & \leq \sum_{|R| \geq 1} |R|A_{ij}(R)p_x(R)
 \leq \sum_{|R| \geq 1} \gamma (B_i(R)+B_j(R)) p_x(R) \nonumber
 \\ &  \leq \sum_{R \subseteq [n]} \gamma (B_i(R)+B_j(R)) p_x(R) = \gamma(\nabla_i F(x)+\nabla_j F(x)).
\end{align}
By adding (\ref{4grad}) and (\ref{ineq2}), we conclude that
\[
||x||_1\sum_{R\subseteq [n]} A_{ij}(R)p_x(R) = ||x||_1 \nabla_{ij}^2 F(x) \leq \max \{2\gamma+1,3\gamma \} (\nabla_i F(x) + \nabla_j F(x)).
\]
\end{proof}
\section{Runtime of the Local Search Algorithm for Meta-Submodular Functions}
\label{app:subseclocaltime}

In this section, we analyze the runtime of the local search algorithm that finds an approximate local optima.

\begin{lemma}
Let $f$ be a non-negative, monotone, $\gamma$-meta-submodular function and $\M=([n],\mathcal{I})$ be a matroid of rank $r$. Let $A\in \mathcal{I}$ be an optimum set, i.e.,
\[
A \in \argmax_{R\in \mathcal{I}} f(R),
\]
and
\[
S_0 \in \argmax_{\{v,v'\}\in\mathcal {I}} f(\{v,v'\}).
\]
Then $f(A)\leq O(r(\gamma+1)^{r-2}) f(S_0)$.
\label{lemma:induction}
\end{lemma}
\begin{proof}
Let $A=\{a_1,\ldots, a_{r}\}$ and $A_i = \{a_1,\ldots,a_i\}$ for $1\leq i \leq r$. By definition of $S_0$ we know that $f(A_2)\leq f(S_0)$. Now by induction we show that for any $2\leq i<j\leq n$, $B_{a_j}(A_i)\leq O((\gamma+1)^{i-1})f(S_0)$. The base case is $i=2$. By definition of $f(S_0)$, monotonicity and meta submodularity of $f$, we have
\begin{align*}
B_{a_j}(A_2) & = B_{a_j}(A_1)+A_{a_2 a_j}(A_1) \leq B_{a_j}(A_1)+\gamma(B_{a_j}(A_1)+B_{a_2}(A_1))\leq (2\gamma+1)f(S_0) \\ & \leq O(\gamma+1)f(S_0).
\end{align*}
Now assume that for $k<j\leq n$, we have $B_{a_j}(A_k)\leq O(\gamma^{k-1})f(S_0)$. We want to show that for $k+1<j\leq n$, we have $B_{a_j}(A_{k+1})\leq O(\gamma^{k})f(S_0)$.
\begin{align*}
B_{a_j}(A_{k+1}) & = B_{a_j}(A_{k})+A_{a_{k+1} a_j}(A_k) \leq B_{a_j}(A_{k})+\frac{\gamma}{k}(B_{a_{k+1}}(A_{k})+B_{a_j}(A_{k})) \\ & \leq (1+\frac{2\gamma}{k}) O((\gamma+1)^{k-1}) f(s_0) \leq O((\gamma+1)^k) f(S_0).
\end{align*}
We know that 
\begin{align*}
f(A)=f(A_2)+\sum_{i=3}^{r}B_{a_i}(A_{i-1}) \leq f(S_0) + \sum_{i=3}^{r} O((\gamma+1)^{i-2})f(S_0) \leq O(r(\gamma+1)^{r-2}) f(S_0)
\end{align*}
\end{proof}

\begin{proposition}\label{prop:localsearchruntime}
Local search algorithm (Algorithm~\ref{alg:localsearch}) runs in $O(n^4(\log(r)+r\log(\gamma+1)/\epsilon)$ time on a $\gamma$-meta submodular functions and a matorid of rank $r$.
\end{proposition}
\begin{proof}
Cost of finding $S_0$ is $O(n^2)$. Also, each iteration of the while loop costs $O(n^2)$. Let $S_k$ be the solution after $k$ iterations and $A$ be an optimum solution. By Lemma~\ref{lemma:induction}, we know
\[
f(S_k)\leq (1+\frac{\epsilon}{n^2})^k f(S_0)\leq f(A)\leq O(r(\gamma+1)^{r-2})f(S_0).
\]
Taking the logarithm, we have
\[
k \ln (1+\frac{\epsilon}{n^2}) \leq O(\ln(r)+(r-2) \ln(\gamma+1)).
\]
Noting that $\frac{x-1}{x}\leq \ln x$ for any $x>0$, we have
\[
k(\frac{\epsilon}{n^2})/(\frac{n^2+\epsilon}{n^2})\leq O(\ln(r)+(r-2) \ln(\gamma+1)).
\]
This yields the result.
\end{proof}

\end{document}